\documentclass[preprint,12pt]{elsarticle}

\usepackage{amssymb}

\usepackage{amsmath,graphicx,amsfonts,amssymb,epsfig,subfigure,mathrsfs,mathtools,cases}
\usepackage{multirow}
\usepackage{multicol}
\usepackage{lipsum}

\usepackage{rotating}
\usepackage{color}
\usepackage{graphicx}%
\usepackage{algorithm}
\usepackage{algpseudocode}
\usepackage{tikz}
\usepackage{stfloats}
\newcommand*\mycirc[1]{%
  \begin{tikzpicture}
    \node[draw,circle,inner sep=1pt] {#1};
  \end{tikzpicture}}

\newtheorem{theorem}{Theorem}[section]

\newtheorem{definition}{Definition}[section]
\newtheorem{lemma}{Lemma}[section]
\newtheorem{proposition}{Proposition}[section]
\newtheorem{remark}{Remark}[section]

\newproof{proof}{Proof}

\journal{arXiv}

\begin{document}

\begin{frontmatter}



\title{
Stability Analysis of Large-Scale Distributed Networked Control Systems with Random Communication Delays: A Switched System Approach
}


\author{Kooktae Lee and Raktim Bhattacharya
\address{Kooktae Lee and Raktim Bhattacharya are with the Department of Aerospace Engineering, Texas A\&M University,
        College Station, TX 77843-3141, USA, {\tt\scriptsize \{animodor,raktim\}@tamu.edu.}
        This research was supported by NSF award \#1349100 with Dr. Almadena Y. Chtchelkanova as the Program Manager.}
}

\begin{abstract}
In this paper, we consider the stability analysis of large-scale distributed networked control systems with random communication delays between linearly interconnected subsystems. The stability analysis is performed in the Markov jump linear system framework. There have been considerable researches on stability analysis of Markov jump systems, however, these methods are not applicable to large-scale systems because large numbers of subsystems result in an extremely large number of the switching modes. To avoid this scalability issue, we propose a new reduced mode model for stability analysis, which is computationally efficient. We also consider the case in which the transition probabilities for the Markov jump process contain uncertainties. We provide a new method that estimates bounds for uncertain Markov transition probability matrix to guarantee the system stability. The efficiency and the usefulness of the proposed methods are verified through examples.
\end{abstract}


\end{frontmatter}



\section{Introduction} 
A networked control system (NCS) is a system that is controlled over a communication network. Recently, NCSs have attracted considerable research interests due to emerging distributed control applications. For example, the NCSs are broadly used in applications including traffic monitoring, networked autonomous mobile agents, chemical plants, sensor networks and distributed software systems in cloud
computing architectures. Due to the communication network between subsystems, communication delays or communication losses may occur, resulting in performance degradation or even instability. 
Therefore, it has led various researches to analyze the NCSs with communication delays \cite{cheng1988distributed}, \cite{walsh2002stability}, \cite{yook2002trading}, \cite{yang2006h}, \cite{nilsson1998stochastic}, \cite{xiao2000control}, \cite{lee2015performance}, \cite{lee2014acc}. In particular, \cite{xiao2000control} constructed a switched system structure for the analysis of NCS by including  actuators, sensors, and the plant as a single system.  

In this paper, we study distributed networked control systems (DNCS) with a large number of spatially distributed linear subsystems (or agents). For such large-scale systems, our primary goal is to analyze system stability when \textit{random communication delays} exist between subsystems.
Typically, such delays have been modeled as \textit{Markov jump linear system} (MJLS) \cite{xiao2000control}, \cite{shi1999control}, \cite{seiler2001analysis}, \cite{zhang2005new}, \cite{shi2009output}, \cite{liu2009stabilization}, in which switching sequence is governed by a Markovian process. Therefore, stability analysis in the existence of communication delays has been performed in the MJLS framework \cite{feng1992stochastic}, \cite{do2006discrete}, \cite{zhang2009stability}. However, these results are applicable to the systems with a small number of switching modes \cite{xiao2000control}, \cite{liu2009stabilization}, \cite{shi2009output}, \cite{seiler2001analysis}, whereas the large-scale DNCSs in which we are particularly interested give rise to an extremely large number of switching modes. 
For such systems, previous approaches for stability analysis are computationally intractable. 
Although \cite{lee2015async} investigated the massively parallel asynchronous numerical algorithm by employing the switched linear system framework  that circumvents the scalability issue with respect to the large number of the switching modes, it is developed for the independent and identically distributed (i.i.d.) switching.
In addition, we are also interested in systems where the transition probabilities are inaccurately known as in \cite{zhang2009stability}, \cite{zhang2010necessary}, \cite{karan2006transition} because, in practice, it is difficult to accurately estimate the Markov transition probability matrix that models the random communication delays.

This paper provides two key contributions to analyze the stability of the large-scale DNCS with random communication delays. Firstly, we guarantee the mean square stability of such systems by introducing a reduced mode model. We prove that the mean square stability for individual switched system implies a necessary and sufficient stability condition for the entire DNCS. This drastically reduces the number of modes necessary for analysis. Secondly, we present a new method to estimate the bound for uncertain Markov transition probability matrix for which stability is guaranteed. These results enable us to analyze large-scale systems in a computationally tractable manner.

Rest of this paper is organized as follows. We introduce the problem for the large-scale DNCS in section 2.
Section 3 presents the switched system framework for the stability analysis with communication delays. In Section 4, we propose the reduced mode model to efficiently analyze stability. Section 5 quantifies the stability region and bound for uncertain Markov transition probability matrix. This is followed by the application of the proposed method to an example system in section 6, and we conclude the paper with section 7.

\textbf{Notation:} The set of real number is denoted by $\mathbb{R}$. The symbols $\parallel\cdot\parallel$ and $\parallel\cdot\parallel_{\infty}$ stand for the Euclidean and infinity norm, respectively. Moreover, the symbol $\#(\cdot)$ denotes the cardinality -- the total number of elements in the given set. Finally, the symbols $\text{tr}(\cdot)$, $\rho(\cdot)$, $\otimes$, and $\textnormal{diag}(\cdot)$ represent trace operator, spectral radius, Kronecker product, and block \textnormal{diag}onal matrix operator, respectively.

\section{Problem Formulation}
\subsection{Distributed networked control system with no delays}
Consider a DNCS with discrete-time dynamics, given by:
\begin{equation}
x_i(k+1) = \sum_{j\in\mathcal{N}_i}A_{ij}x_j(k),\quad i=1,2,\hdots,N,\label{eqn:1}
\end{equation}
where $k$ is a discrete-time index, $N$ is the total number of agents (subsystems), $x_i\in \mathbb{R}^n$ is a state for the $i^{th}$ agent, $\mathcal{N}_i$ is a set of neighbors for $x_i$ including the agent $x_i$ itself, and $A_{ij}\in\mathbb{R}^{n\times n}$ is a time-invariant system matrix that represents the linear interconnections between agents. Note that we have $A_{ij} = 0$ if there is no interconnection between the agents $i$ and $j$. 

To represent the entire systems dynamics, we define the state $x(k)\in\mathbb{R}^{Nn\times Nn}$ as $x(k)\triangleq [x_1(k)^{\top},x_2(k)^{\top},\hdots,x_N(k)^{\top}]^{\top}$.
Then, the system dynamics of the DNCS is given as
\begin{align}
x(k+1) = \mathcal{A}x(k),\label{eqn:2}
\end{align}
where the matrix $\mathcal{A}\in\mathbb{R}^{Nn\times Nn}$ is defined by
\begin{align*}
&\quad \mathcal{A} \triangleq \begin{bmatrix}
A_{11} & A_{12} & A_{13} & \cdots & A_{1N}\\
A_{21} & A_{22} & A_{23} & \cdots & A_{2N}\\
A_{31} & A_{32} & A_{33} & \cdots & A_{3N}\\
\vdots& \vdots & \vdots & \ddots &\vdots\\
A_{N1} & A_{N2} & A_{N3} & \cdots & A_{NN}
\end{bmatrix}, \\
A_{ij} &=
\begin{dcases}
0, \text{ if no connection between the agents $i$ and $j$,}\\
A_{ij}, \text{ otherwise.}
\end{dcases}
\end{align*}

For the discrete-time system in \eqref{eqn:2}, it is well known that the system is stable if and only if the condition $\rho(\mathcal{A}) < 1$ is satisfied. We assume that the system in \eqref{eqn:2}, which is the case without communication delays is stable throughout the paper. Then, we address the problem to analyze the stability in the presence of \textit{random communication delays}. We remind the reader that $N$ is very large.

\subsection{DNCS with communication delays}
Often, network communication between agents encounter time delays or packet losses while sending and receiving data as described in Fig. \ref{fig.1}. We denote the symbol $\tau$ as communication delays and assume that $\tau$ has a discrete value bounded by  $0\leq \tau \leq \tau_d < \infty$, where $\tau_d$ is a finite-valued maximum delay. Then, the dynamics for the agent $i$  with communication delays can be expressed as:
\begin{figure}
\centering
\includegraphics[scale=0.6]{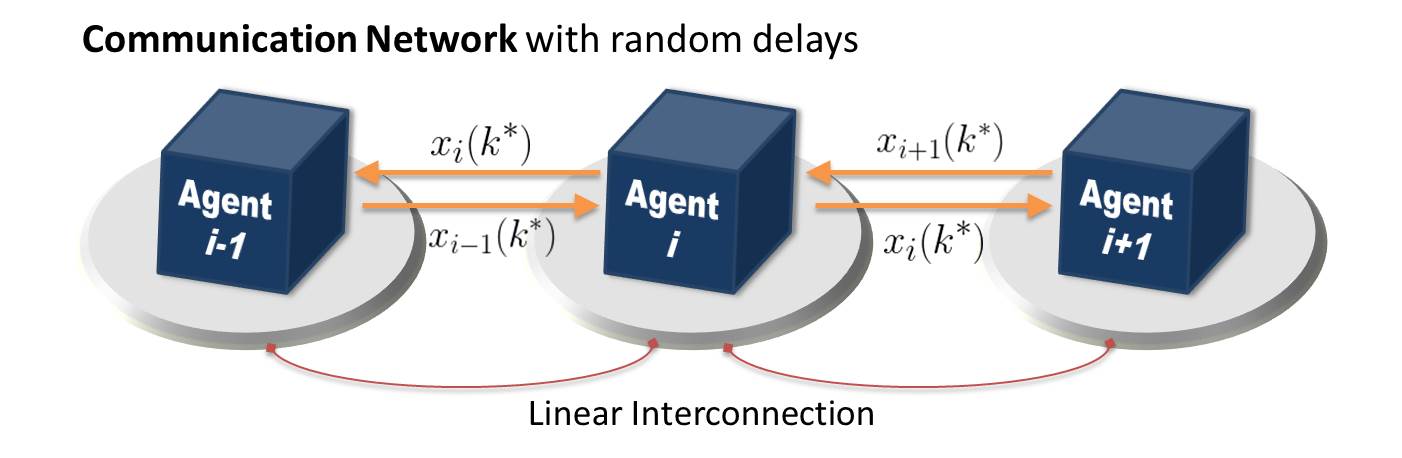}
\caption{The schematic of the large-scale distributed networked control system. The communication delay is represented by $k^{*}\triangleq k - \tau$, where $\tau$ is the random communication delay term.}\label{fig.1}
\end{figure}
\begin{align}
&x_i(k+1) = \sum_{j\in\mathcal{N}_i}A_{ij}x_j(k^*),\quad i=1,2,\hdots,N,\label{eqn:3}
\end{align}
where $k^* \triangleq k-\tau$.
Note that we have no communication delays when $i=j$ since there is no communication in this case.

The communication delay, modeled as a stochastic process, is represented by the term $k^*$. To analyze the stability of the DNCS, we define an augmented state $X(k)$ as $X(k) \triangleq [x(k)^{\top},\:x(k-1)^{\top},\:\cdots,\: x(k-\tau_d)^{\top}]^{\top}\in\mathbb{R}^{Nnq \times Nnq}$, where $q\triangleq \tau_d+1$. Then, the dynamics for the entire system is given by
\begin{align}
X(k+1) = W(k)X(k),\label{eqn:4}
\end{align}
where 
$W(k) \triangleq \begin{bmatrix}
\tilde{A}_1(k) & \tilde{A}_2(k) & \cdots & \tilde{A}_{q-1}(k) & \tilde{A}_q(k)\\
I & 0 & \cdots & 0 & 0\\
0 & I & \cdots & 0 & 0\\
\vdots & \vdots & \ddots & \vdots &\vdots \\
0 & 0 & \cdots & I & 0\\
\end{bmatrix}\in \mathbb{R}^{Nnq\times Nnq},$\\
the matrix $I$ denotes an identity matrix with proper dimensions, and the time-varying matrices $\tilde{A}_j(k)\in\mathbb{R}^{Nn\times Nn}$, $j=1,2,\hdots,q$, model the randomness in the communication delays between neighboring agents.

\section{Switched System Approach}
Without loss of generality, the dynamics of the large-scale DNCS with communication delays in \eqref{eqn:4} can be transformed into a switched system framework as :
\begin{equation}
x(k+1) = W_{\sigma(k)}x(k), \quad \sigma(k)\in\{1,2,\cdots,m\},\label{eqn:5}
\end{equation}
where the set of matrices $\{W_{\sigma(k)}\}_{\sigma(k)=1}^{m}$ represents all possible communications delays between interconnected agents, $\{\sigma(k)\}$ is the switching sequence, and $m$ is the total number of switching modes. 
When the switching sequence $\{\sigma(k)\}$ is stochastic, \eqref{eqn:5} is referred to as a stochastic switched linear system or a stochastic jump linear system \cite{lee2015performance}.
For the stochastic switched linear system, the switching sequence $\{\sigma(k)\}$ is governed by the mode-occupation switching probability $\pi(k)=[\pi_1(k),\pi_2(k),\hdots,$ $\pi_m(k)]$, where $\pi_i$ is a fraction number, representing the modal probability such that $ \sum_{i=1}^{m}\pi_i=1$ and $0 \leq \pi_i \leq 1$, $\forall i$. 
Typically, randomness in communication delays or communication losses has been modeled by the MJLS framework \cite{shi1999control}, \cite{seiler2001analysis}, \cite{zhang2005new}, \cite{shi2009output}, \cite{liu2009stabilization}.
Therefore, we make the following assumption in our analysis.
\begin{itemize}
\item Assumption: Consider the stochastic jump linear system \eqref{eqn:5} with the switching probability $\pi(k)=[\pi_1(k),\pi_2(k),\hdots,\pi_m(k)]$. Then, $\pi(k)$ is updated by the Markovian process given by $\pi(k+1) = \pi(k)P$, where $P\in\mathbb{R}^{m\times m}$ is the Markov transition probability matrix.
\end{itemize}
Since the MJLS is a family of the stochastic switched linear system, various stability notions can be defined \cite{feng1992stochastic}. In this paper, we will consider the mean square stability condition, defined below.

\begin{definition}(Definition 1.1 in \cite{fang2002stabilization})
The MJLS is said to be mean square stable if for any initial condition $x_0$ and arbitrary initial probability distribution $\pi(0)$, $\displaystyle\lim_{k\rightarrow\infty}\mathbb{E}\left[||x(k,x_0)||^2\right]=0$.
\end{definition}

Note that for the large-scale DNCS, the total number of switching modes $m$ depends on the size $q$ and $N$. 
Since the communication delays take place independently while receiving and sending the data for each agent, $m$ is calculated by counting all possible scenarios to distribute every matrices $A_{ij}\in\mathbb{R}^{n\times n}$ for $i\neq j$ in the block matrix $\mathcal{A}\in\mathbb{R}^{Nn\times Nn}$ given in \eqref{eqn:2}, into each $\tilde{A}_j(k)\in\mathbb{R}^{Nn\times Nn}$, $j=1,2,\hdots,q,$ given in \eqref{eqn:4}, which  results in $m=q^{N(N-1)}$. For large $N$, $m$ is quite large, which makes current analysis tools for the MJLS computationally intractable.

Before we further proceed, we introduce the following proposition that was developed for the stability analysis of the MJLS.

\begin{proposition}\label{prop:4.1}(Theorem 1 in \cite{costa1993stability})
The MJLS with the Markov transition probability matrix $P$ is mean square stable if and only if
\begin{eqnarray}
\rho\left(\left(P^{\top}\otimes I\right)\textnormal{diag}(W_j\otimes W_j) \right) < 1,\label{eqn:10}
\end{eqnarray}
where $I$ is an identity matrix with a proper dimension,
{\small
\begin{align*}
\textnormal{diag}(W_j\otimes W_j) \triangleq \begin{bmatrix}
(W_1\otimes W_1) & 0 & 0 & \cdots & 0\\
0 & (W_2\otimes W_2) & 0 & \cdots & 0\\
\vdots &  & \ddots &&\vdots\\
 0 & 0 &  & (W_{m\text{-}1}\otimes W_{m\text{-}1}) & 0\\
0 & 0 & \hdots & 0 & (W_m\otimes W_m)
\end{bmatrix},
\end{align*}}
and $m$ is the total number of the switching modes.
\end{proposition}

For the given set of matrices $\{W_{\sigma(k)}\}_{\sigma(k)=1}^{m}$ and the transition probability matrix $P$, one can always compute the spectral radius given in \eqref{eqn:10}, and hence guarantee the system stability.

Unfortunately, this condition is not applicable to large-scale DNCSs since $N$ is very high and results in extremely large $m$. 
For example, even if $q=2$ and $N=100$, we have $m = 2^{100\times 99}$. It is not possible to compute the spectral radius for such problems. To circumvent this scalability issue, we present next a new analysis approach for such large-scale DNCSs.

\section{Stability with Reduced Mode Dynamics}
In this section, we define a new augmented state to reduce the mode numbers as follows:
\begin{align*}
\hat{x}_i(k) \triangleq [\tilde{x}_i(k)^{\top}, \: \tilde{x}_i(k-1)^{\top}, \: \cdots, \: \tilde{x}_i(k-\tau_d)^{\top}]^{\top}\in\mathbb{R}^{\hat{n}_in q},
\end{align*}
where $\tilde{x}_i(k) \triangleq [x_i(k)^{\top}, \:\: x_j(k)^{\top}]^{\top}\in\mathbb{R}^{\hat{n}_in}$, $\hat{n}_i\triangleq\#(\mathcal{N}_i)$, and $x_j(k)\in\mathbb{R}^n$ with $j\in\mathcal{N}_i$ denotes all states that are neighbor to $x_i(k)\in\mathbb{R}^n$.

Then, we can construct a switched linear system framework similarly to \eqref{eqn:5} as follows:
\begin{eqnarray}
\hat{x}_i(k+1) = \hat{W}^i_{\sigma_i(k)}\hat{x}_i(k),\quad \sigma_i(k)\in\{1,2,\hdots,m_i\},\label{eqn:12}
\end{eqnarray}
where 
$\hat{W}^i_{\sigma_i(k)} \triangleq \begin{bmatrix}
\hat{A}_{1}(k) & \hat{A}_{2}(k) & \cdots & \hat{A}_{q-1}(k) & \hat{A}_{q}(k)\\
I & 0 & \cdots & 0 & 0\\
0 & I & \cdots & 0 & 0\\
\vdots & \vdots & \ddots & \vdots & \vdots\\
0 & 0 & \cdots & I & 0
\end{bmatrix} \in \mathbb{R}^{\hat{n}_in q\times \hat{n}_in q}$\\
with the time-varying matrix $\hat{A}_j(k)\in\mathbb{R}^{\hat{n}_in\times \hat{n}_in}$, $j=1,2,\hdots,q$. In this case, the total number of the switching modes for \eqref{eqn:12} is given by 
$m_i=q^{\hat{n}_i(\hat{n}_i-1)}$.

By implementing the reduce mode model given in \eqref{eqn:12}, we will provide a computationally efficient tool for the stability analysis of the original DNCS in the following theorem.

\begin{theorem}\label{theorem:4.1}
Consider the large-scale DNCS \eqref{eqn:5} with Markovian communication delays accompanied by the transition probability matrix $P$. The necessary and sufficient condition for the mean square stability of this system is then given by
\begin{eqnarray}
\rho\Big(({P^i}^{\top}\otimes I)\text{\textnormal{diag}}(\hat{W}^i_j\otimes \hat{W}^i_j)\Big) < 1, \quad\forall i=1,2,\hdots,N,\label{eqn:13}
\end{eqnarray}
where $P^i\in\mathbb{R}^{m_i\times m_i}$ is the transition probability matrix for the reduced mode MJLS given in \eqref{eqn:12}, $I$ is an identity matrix with a proper dimension, $N$ is the total number of the agents in the system, $m_i=q^{\hat{n}_i(\hat{n}_i-1)}$ is the total mode numbers for the reduce mode MJLS, and
{\footnotesize
\begin{align*}
\textnormal{diag}(\hat{W}^i_j\otimes \hat{W}^i_j) \triangleq \begin{bmatrix}
(\hat{W}^i_1\otimes \hat{W}^i_1) & 0 & 0 & \cdots & 0\\
0 & (\hat{W}^i_2\otimes \hat{W}^i_2) & 0 & \cdots & 0\\
\vdots &  & \ddots &&\vdots\\
 0 & 0 &  & (\hat{W}^i_{m_i\text{-}1}\otimes \hat{W}^i_{m_i\text{-}1}) & 0\\
0 & 0 & \hdots & 0 & (\hat{W}^i_{m_i}\otimes \hat{W}^i_{m_i})
\end{bmatrix}.
\end{align*}}
\end{theorem}
\begin{proof}
Let the matrix $Q^{i}(k)$ be of the form $Q^{i}(k)\triangleq \mathbb{E}[\hat{x}_i(k)\hat{x}_i(k)^{\top}]$. Then, $Q^i(k)$ is alternatively obtained by the following equation: $\displaystyle Q^{i}(k) = \sum_{s=1}^{m_i}Q_{s}^{i}(k)$, where $\displaystyle Q_{s}^{i}(k) \triangleq \mathbb{E}\left[\hat{x}_i(k){\hat{x}_i(k)}^{\top}|\sigma_i(k)=s\right]\pi_s^{i}(k)$, and $\pi_{s}^{i}(k) \triangleq \mathbf{Pr}\big(\sigma_i(k)=s\big)$. Then, $Q_s^i(k)$ satisfies
{\small
\begin{align*}
Q_s^i(k) &= \sum_{r=1}^{m_i}\mathbb{E}[\hat{x}_i(k){\hat{x}_i(k)}^{\top}\mid \sigma_i(k) =s,\sigma_i(k-1) =r]\\[-0.1in]
&\qquad\qquad\qquad\qquad\qquad\qquad\mathbf{Pr}(\sigma_i(k-1) = r\mid \sigma_i(k) =s)\pi_s^i(k)\\
&= \sum_{r=1}^{m_i}\mathbb{E}[\hat{x}_i(k){\hat{x}_i(k)}^{\top}\mid \sigma_i(k) =s,\: \sigma_i(k-1) =r]\\[-0.1in]
&\qquad\qquad\qquad\qquad\qquad\qquad\underbrace{\mathbf{Pr}(\sigma_i(k) = s\mid \sigma_i(k-1) =r)}_{\triangleq p^{i}_{rs}}\pi_r^i(k-1)\\
&= \sum_{r=1}^{m_i}{p}_{rs}^{i}
\:\mathbb{E}[\hat{x}_i(k){\hat{x}_i(k)}^{\top}\mid \sigma_i(k) =s, \sigma_i(k-1) =r]\pi_r^i(k-1)\\
&= \sum_{r=1}^{m_i}{p}_{rs}^{i}\:\mathbb{E}[\hat{W}^i_{\sigma_i(k-1)}\hat{x}_i(k-1){\hat{x}_i(k-1)}^{\top}\hat{W}^{i^{\top}}_{\sigma_i(k-1)} \mid \sigma_i(k-1) =r]\pi_r^i(k-1)\\
&= \sum_{r=1}^{m_i}{p}_{rs}^{i}\:\hat{W}^i_{r}\underbrace{\mathbb{E}[\hat{x}_i(k-1){\hat{x}_i(k-1)}^{\top}\mid \sigma_i(k-1) =r]\pi_r^i(k-1)}_{=Q_r^i(k-1)}\hat{W}^{i^{\top}}_{r}\\
&= \sum_{r=1}^{m_i}{p}_{rs}^{i}\:\hat{W}^i_{r}Q_r^i(k-1){\hat{W}^{i^{\top}}_{r}}.
\end{align*}}
In the second equality of above equation, $p_{rs}^i$ denotes the mode transition probability from $r$ to $s$ in the Markov transition probability matrix $P^i$.

Taking the vectorization in above equation results in
\begin{align*}
\text{vec}\left(Q_s^i(k)\right) &=  \text{vec}\left(\sum_{r=1}^{m_i}{p}_{rs}^{i}\:\hat{W}^i_{r}Q_r^i(k-1){\hat{W}^{i^{\top}}_{r}}\right)\\
&= \sum_{r=1}^{m_i}{p}_{rs}^{i}\text{vec}\left(\hat{W}^i_{r}Q_r^i(k-1){\hat{W}^{i^{\top}}_{r}}\right)\\
&= \sum_{r=1}^{m_i}{p}_{rs}^{i}(\hat{W}_r^i\otimes \hat{W}_r^i)\text{vec}(Q_r^i(k-1)).
\end{align*}
In the last equality, we used the property that $\text{vec}(ABC) = (C^{\top}\otimes A)\text{vec}(B)$. We define a new variable $y_{(\cdot)}^i(k) \triangleq \text{vec}\left(Q_{(\cdot)}^i(k)\right)$, which leads to
\begin{align*}
&y_s^i(k) = \sum_{r=1}^{m_i}{p}_{rs}^{i}(\hat{W}^i_{r}\otimes \hat{W}^i_{r})y_r^i(k-1).
\end{align*}


By stacking $y_{(\cdot)}^i(k)$ from $1$ up to $m_i$, with a new definition for the augmented state $\hat{y}^i(k) \triangleq [{y_1^i}(k)^{\top}\:{y_2^i}(k)^{\top}\:\hdots\:{y_{m_i}^i}(k)^{\top}]^{\top}$, we have the following recursion equation:
{\footnotesize
\begin{align*}
\hat{y}^i(k)&=\underbrace{\begin{bmatrix}
{p}_{11}^i(\hat{W}^i_{1}\otimes \hat{W}^i_{1}) & {p}_{21}^i(\hat{W}^i_{2}\otimes \hat{W}^i_{2}) & \hdots & {p}_{m_i1}^i(\hat{W}^i_{m_i}\otimes \hat{W}^i_{m_i})\\
{p}_{12}^i(\hat{W}^i_{1}\otimes \hat{W}^i_{1}) & {p}_{22}^i(\hat{W}^i_{2}\otimes \hat{W}^i_{2}) & \hdots & {p}_{m_i2}^i(\hat{W}^i_{m_i}\otimes \hat{W}^i_{m_i})\\
\vdots & \vdots & \ddots & \vdots\\
{p}_{1m_i}^i(\hat{W}^i_{1}\otimes \hat{W}^i_{1}) & {p}_{2m_i}^i(\hat{W}^i_{2}\otimes \hat{W}^i_{2}) & \hdots & {p}_{m_im_i}^i(\hat{W}^i_{m_i}\otimes \hat{W}^i_{m_i})\\
\end{bmatrix}}_{
=({P^i}^{\top}\otimes I)\text{\textnormal{diag}}(\hat{W}^i_j\otimes \hat{W}^i_j)}
\underbrace{
\begin{bmatrix}
y_1^i(k-1)\\
y_2^i(k-1)\\
\vdots\\
y_{m_i}^i(k-1)
\end{bmatrix}}_{=\hat{y}^i(k-1)}.
\end{align*}}

From the above equation, it is clear that $\rho\Big(({P^i}^{\top}\otimes I)\text{\textnormal{diag}}(\hat{W}^i_j\otimes \hat{W}^i_j)\Big) < 1$ implies $\displaystyle \lim_{k\rightarrow\infty}\hat{y}^i(k)=0$, and hence this leads to $\displaystyle\lim_{k\rightarrow\infty}Q^i(k) = 0 \Longleftrightarrow \lim_{k\rightarrow\infty}\text{tr}\left(Q^i(k)\right)=0 \Longleftrightarrow \lim_{k\rightarrow\infty}\mathbb{E}\left[||\hat{x}_i(k)||^2\right]=0$, which is the sufficient mean square stability condition for $\hat{x}_i(k)$. On the other hand, if we have $\rho\Big(({P^i}^{\top}\otimes I)\text{\textnormal{diag}}(\hat{W}^i_j\otimes \hat{W}^i_j)\Big) > 1$, then $\hat{y}^i(k)$ will diverge, resulting in necessity for the mean square stability of $\hat{x}_i(k)$.
Hence, the spectral radius being less than one is the necessary and sufficient mean square stability condition for the state $\hat{x}_i(k)$.
Further, we have $\displaystyle\lim_{k\rightarrow \infty}\mathbb{E}\left[||\hat{x}_i(k)||^2\right]=0,\: \forall i=1,2,\hdots,N \Longleftrightarrow \lim_{k\rightarrow \infty}\mathbb{E}\left[||x(k)||^2\right]=0,$ where $x(k)$ is the state for the DNCS defined in \eqref{eqn:5}.
This concludes the proof.
\end{proof}
\begin{remark}
Theorem \ref{theorem:4.1} provides an efficient way to analyze the stability for the large-scale DNCSs. 
The key idea stems from the hypothesis that the stability of each subsystem by partitioning the original system will provide the stability of the entire system.
Without any relaxation or conservatism, theorem \ref{theorem:4.1} proved the necessary and sufficient condition for stability, which is equivalent to \eqref{eqn:10} for the mean square stability of the entire system. Compared to the total number of modes of full state model \eqref{eqn:5}, which is $q^{N(N-1)}$, the reduced mode model \eqref{eqn:12} has total $\sum_{i=1}^{N}q^{\hat{n}_i(\hat{n}_i-1)}$ modes. Consequently, the growth of mode numbers in full state model is exponential with respect to $N^2$, whereas that in reduced mode model is \textbf{linear} with regard to $N$. Therefore, theorem \ref{theorem:4.1} is computationally more efficient. 
\end{remark}

\section{Stability Region and Stability Bound for Uncertain Markov Transition Probability Matrix}
The Markov transition probability matrix can be obtained from data of communication delays. However, the statistics itself contains uncertainties due to the uncertainty in the data. Thus, one cannot estimate the exact transition probability matrix in practice. In this subsection, we assume that the Markov transition probability matrix has uncertainty, i.e. $P^i = \bar{P}^i + \Delta P^i$, where $\bar{P}^i$ is the nominal value and $\Delta P^i$ is the uncertainty in the Markov transition probability matrix for $i^{th}$ subsystem. Due to the variation in $\Delta P^i$, the system stability may change and hence we want to estimate the bound for $\Delta P^i$ to guarantee the system stability. 
Here we assume that $\Delta P^i$ has the following structure:
\begin{align}
\Delta P^i \triangleq \begin{bmatrix}
\Delta p^i_{11} & \Delta p^i_{12} & \cdots & \Delta p^i_{1m_i}\\
\Delta p^i_{21} & \Delta p^i_{22} & \cdots & \Delta p^i_{2m_i}\\
\vdots & \vdots & \ddots & \vdots\\
\Delta p^i_{m_i1} & \Delta p^i_{m_i2} & \cdots & \Delta p^i_{m_im_i}
\end{bmatrix}, \in\mathbb{R}^{m_i\times m_i}\nonumber\\
\text{ s.t. }\sum_{s=1}^{m_i}\Delta p^i_{rs} = 0, \forall r=1,2,\hdots,m_i\label{eqn:11-1}
\end{align}

Since we have a constraint such that the row sum has to be zero for $\Delta P^i$ in above equation, we aim to find the feasible maximum bound for each row, $\varepsilon^{i}_{r}$, satisfying the inequality $|\Delta p^i_{rs}| \leq \varepsilon^i_{r},\:\forall r=1,2,\hdots,m_i$, to guarantee the system stability. Then, $\varepsilon^i_r$, $\forall r=1,2,\hdots,m_i$ can be obtained by the following two steps.\\

\noindent\textbf{Step 1: Solve via Linear Programming (LP)}\\[-0.2in]
\begin{eqnarray}
&\text{maximize}&\mathbf{1}^{\top}z\qquad\text{(for upper bound)}\label{eqn:14}\\
\Big(\text{or} &\text{minimize}&\mathbf{1}^{\top}z\qquad\text{(for lower bound)}\Big)\nonumber\\
&\text{subject to}&\mathbb{A}|z| < b_s,\quad \forall s=1,2,\hdots, m_i\label{eqn:15}\\
&& lb_s \leq z_s \leq ub_s, \forall s=1,2,\hdots, m_i\label{eqn:16}
\end{eqnarray}
where
\begin{align*}
&z_s \triangleq [\Delta p^i_{1s},\: \Delta p^i_{2s},\:\cdots,\: \Delta p^i_{m_is}]^{\top},&\qquad\qquad\qquad\qquad\qquad\qquad\qquad\\
&z \triangleq [z_1^{\top},\: z_2^{\top},\: \cdots, \: z_{m_i}^{\top}]^{\top},
\end{align*}
\begin{align*}
&\mathbb{A} \triangleq \Big[
\alpha_1,\:\alpha_2,\:\cdots,\:\alpha_{m_i}
\Big],\text{ with }
\alpha_j \triangleq \parallel \hat{W}^i_j \otimes \hat{W}^i_j \parallel_{\infty},\: j=1,2,\hdots,m_i,\\
&b_s \triangleq 1 -\sum_{r=1}^{m_i}\alpha_r\bar{p}_{rs}^{i},\\
&lb_s \triangleq [-\bar{p}^i_{1s}, \: -\bar{p}^i_{2s}, \: \cdots \: -\bar{p}^i_{m_is}]^{\top},\\
&ub_s \triangleq [1-\bar{p}^i_{1s}, \: 1-\bar{p}^i_{2s}, \: \cdots \: 1-\bar{p}^i_{m_is}]^{\top}.
\end{align*}

The inequality constraint \eqref{eqn:15} in the LP problem guarantees the mean square stability according to the Lemma \ref{lemma:4.1} and Theorem \ref{theorem:4.2}. The term $lb_s$ and $ub_s$ in \eqref{eqn:16} are the lower and upper bounds for $z_s$, according to $0 \leq p^i_{rs}=(\bar{p}^i_{rs} + \Delta p^i_{rs}) \leq 1$.\\
\begin{figure}[h]
\centering
\includegraphics[scale=0.75]{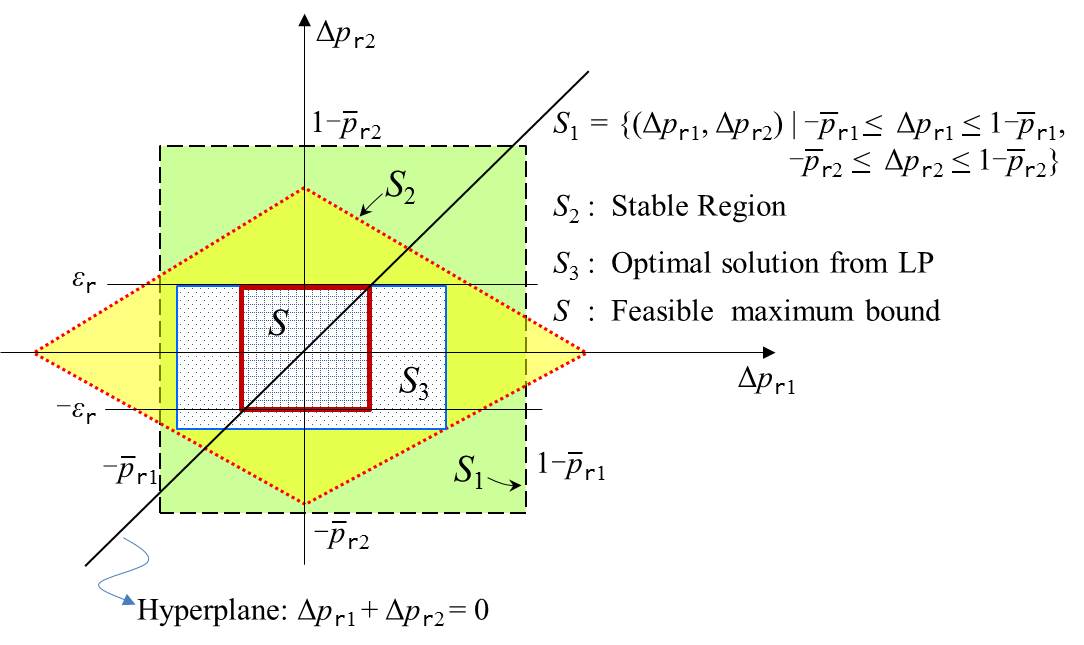}
\caption{The geometry of the Stability Region Analysis for the uncertain Markov transition probability matrix when $m_i=2$. Each region is described in the figure.}\label{fig:2}
\end{figure}

\noindent\textbf{Step 2: Obtain Feasible Solution with Hyperplane Constraint}\\
We can compute the feasible maximum bound for $\Delta p^i_{rs}$ as follows.
\begin{align}
&\varepsilon_r^i = \min\left(\min(|\varepsilon_{r,\text{lb}}^{i}|),\min(\:|\varepsilon_{r,\text{ub}}^i|)\right),\quad r = 1,2,\hdots,m_i.\label{eqn:17}
\end{align}
where 
$\varepsilon_{r,\text{lb}}^i \triangleq [(\Delta p^i_{r1})^{\star}_{\text{lb}},\:(\Delta p^i_{r2})^{\star}_{\text{lb}},\:\hdots,\: (\Delta p^i_{rm})^{\star}_{\text{lb}}]^{\top}$, $\varepsilon_{r,\text{ub}}^i \triangleq [(\Delta p^i_{r1})^{\star}_{\text{ub}},\:(\Delta p^i_{r2})^{\star}_{\text{ub}},$ $\hdots,\: (\Delta p^i_{rm})^{\star}_{\text{ub}}]^{\top}$, and $(\Delta p_{rs}^i)^{\star}_{lb}$, $(\Delta p_{rs}^i)^{\star}_{ub}$ denote optimal lower and upper bounds for $\Delta p_{rs}^i$, obtained from the LP, respectively.

Since upper or lower bounds are solved by maximizing or minimizing the objective function, $(\Delta p^i_{rs})^{\star}$ has different values for upper and lower bounds. 
Fig.\ref{fig:2} shows the geometry of stability region analysis for uncertain transition probability matrix. The region $S_1$ stands for the bounds that come from $-\bar{p}^i_{rs} \leq \Delta p^i_{rs} \leq 1-\bar{p}^i_{rs}$. $S_2$ can be obtained from inequality constraint \eqref{eqn:15}. The region $S_3$ denotes the solution from the LP and $S$ is the feasible maximum upper and lower bounds with a stability guarantee. 
Note that $\Delta P^i$ satisfies $\sum_{s=1}^{m_i}\Delta p^i_{rs} = 0, \: \forall r=1,2,\hdots,m_i$ and hence feasible solutions should lie on the hyperplane, satisfying $\Delta p^i_{r1} + \Delta p^i_{r2} + \hdots + \Delta p^i_{rm_i} = 0$, $\forall r=1,2,\hdots,m_i$. Therefore, we can compute the feasible maximum bound from \eqref{eqn:17} for each row $r$.

Now we prove that inequality constraint \eqref{eqn:15} guarantees the system stability.
\begin{lemma}\label{lemma:4.1}
Consider a block matrix $X$ defined by
\begin{align*}
X = \begin{bmatrix}
X_{11} & X_{12} & \cdots & X_{1m}\\
X_{21} & X_{22} & \cdots & X_{2m}\\
\vdots & \vdots& \vdots & \vdots\\
X_{m1} & X_{m2} & \cdots & X_{mm}
\end{bmatrix},
\end{align*}
where matrix $X_{ij}\in\mathbb{R}^{n\times n}$.
Then, we have $\rho\left(X\right) < 1$, if $\displaystyle\sum_{j=1}^{m}\left\lVert X_{ij}\right\rVert_{\infty} < 1,\:\forall i=1,2,\hdots,m.$
\end{lemma}
\begin{proof}
For the block matrix $X$ given above, the following inequality condition $\parallel X \parallel_{\infty} \leq \max_i \sum_{j=1}^{m}$ $\parallel X_{ij} \parallel_{\infty}$ holds. Also, it is well known that $\rho(X) \leq$ $\parallel X \parallel_p$ for any choice of $p$. 

Therefore, we conclude that $\sum_{j=1}^{m}\parallel X_{ij} \parallel_{\infty} < 1,\:\forall i=1,2,\hdots,m \Rightarrow \rho(X) \leq \parallel X \parallel_{\infty} < 1$.
\end{proof}
\vspace{0.1in}

\begin{theorem}\label{theorem:4.2}
Consider the MJLS \eqref{eqn:5} for the large-scale DNCS with communication delays. Then, \eqref{eqn:5} is mean square stable if
\begin{eqnarray*}
&\displaystyle\sum_{r=1}^{m_i}\alpha_r |\Delta p^i_{rs}| < \beta_s,\quad\begin{matrix}
\forall s=1,2,\hdots,m_i,\\
\forall i=1,2,\hdots,N
\end{matrix}
\end{eqnarray*}
where $\alpha_r = \parallel \hat{W}^i_r \otimes \hat{W}^i_r\parallel_{\infty}$ and $\displaystyle \beta_s = 1 - \sum_{r=1}^{m_i}\bar{p}^i_{rs}\parallel \hat{W}^i_r \otimes \hat{W}^i_r\parallel_{\infty}$, 
is satisfied.
\end{theorem}

\begin{proof}
If the Markov transition probability matrix for the system in \eqref{eqn:12} has the uncertainty denoted by $P^i = \bar{P}^i + \Delta P^i$, then 
the term $\rho\Big(({P^i}^{\top}\otimes I)$ $\text{\textnormal{diag}}(\hat{W}_{j}^{i}\otimes \hat{W}_{j}^{i})\Big)$ in \eqref{eqn:13} can be expressed as
\begin{align}
&\rho\Big(({P^i}^{\top}\otimes I)\text{\textnormal{diag}}(\hat{W}_{j}^{i}\otimes \hat{W}_{j}^{i})\Big)\nonumber\\
= &\rho\Big(\big({(\bar{P}^i + \Delta P^i)}^{\top}\otimes I\big)\text{\textnormal{diag}}(\hat{W}_{j}^{i}\otimes \hat{W}_{j}^{i})\Big)\nonumber\\
= &\rho\bigg(\Big(({\bar{P^i}}^{\top}\otimes I) + (\Delta {P^i}^{\top}\otimes I)\Big)\text{\textnormal{diag}}(\hat{W}_{j}^{i}\otimes \hat{W}_{j}^{i})\bigg)\nonumber\\
= &\rho\bigg(({\bar{P^i}}^{\top}\otimes I)\text{\textnormal{diag}}(\hat{W}_{j}^{i}\otimes \hat{W}_{j}^{i}) + (\Delta {P^i}^{\top}\otimes I)\text{\textnormal{diag}}(\hat{W}_{j}^{i}\otimes \hat{W}_{j}^{i})\bigg)\nonumber\\
\leq &\parallel({\bar{P^i}}^{\top}\otimes I)\text{\textnormal{diag}}(\hat{W}_{j}^{i}\otimes \hat{W}_{j}^{i}) + (\Delta {P^i}^{\top}\otimes I)\text{\textnormal{diag}}(\hat{W}_{j}^{i}\otimes \hat{W}_{j}^{i})\parallel_{\infty}\nonumber\\
\leq &\parallel({\bar{P^i}}^{\top}\otimes I)\text{\textnormal{diag}}(\hat{W}_{j}^{i}\otimes \hat{W}_{j}^{i})\parallel_{\infty} + \parallel(\Delta {P^i}^{\top}\otimes I)\text{\textnormal{diag}}(\hat{W}_{j}^{i}\otimes \hat{W}_{j}^{i})\parallel_{\infty},\label{eqn:15-1}
\end{align}
In the first inequality, we used the fact that $\rho(\cdot) \leq \parallel\cdot\parallel_{\infty}$ and the sub-multiplicative property was applied in the last inequality. The block matrix structure for each term of the last inequality is alternatively expressed as follows:

{\small
\begin{align*}
&\left\lVert(\bar{P^i}^{\top}\otimes I)\text{\textnormal{diag}}(\hat{W}_{j}^{i}\otimes \hat{W}_{j}^{i})\right\rVert_{\infty}\\
&=\left\lVert
\underbrace{\begin{bmatrix}
\bar{p}^i_{11}I & \bar{p}^i_{21}I & \cdots & \bar{p}^i_{m_i1}I\\
\bar{p}^i_{12}I & \bar{p}^i_{22}I & \cdots & \bar{p}^i_{m_i2}I\\
\vdots & \vdots & \cdots & \vdots \\
\bar{p}^i_{1m_i}I & \bar{p}^i_{2m_i}I & \cdots & \bar{p}^i_{m_im_i}I\\
\end{bmatrix}}_{=(\bar{P^i}^{\top}\otimes I)}
\underbrace{\begin{bmatrix}
\hat{W}^i_1\otimes \hat{W}^i_1 & 0 & \cdots & 0\\
0 & \hat{W}^i_2\otimes \hat{W}^i_2 &  & 0\\
\vdots & & \ddots & \\
0 & 0 & &\hat{W}^i_{m_i}\otimes \hat{W}^i_{m_i}
\end{bmatrix}}_{=\text{\textnormal{diag}}(\hat{W}_{j}^{i}\otimes \hat{W}_{j}^{i})}
\right\rVert_{\infty}\\
\\
&=\left\lVert \begin{matrix}
\gamma_1\bar{p}^i_{11} & \gamma_2\bar{p}^i_{21} & \cdots & \gamma_{m_i}\bar{p}^i_{m_i1}\\
\gamma_1\bar{p}^i_{12} & \gamma_2\bar{p}^i_{22} & \cdots & \gamma_{m_i}\bar{p}^i_{m_i2}\\
\vdots & \vdots & \ddots & \vdots\\
\gamma_1\bar{p}^i_{1m_i} & \gamma_2\bar{p}^i_{2m_i} & \cdots & \gamma_{m_i}\bar{p}^i_{m_im_i}\\
\end{matrix} \right\rVert_{\infty},
\end{align*}}
where $\gamma_j = (\hat{W}_j^i \otimes \hat{W}_j^i)$, $j=1,2,\hdots,m_i$, and similarly,
\begin{align*}
&\left\lVert(\Delta {P^i}^{\top}\otimes I)\text{\textnormal{diag}}(\hat{W}_{j}^{i}\otimes \hat{W}_{j}^{i})\right\rVert_{\infty} =
\left\lVert \begin{matrix}
\gamma_1\Delta p^i_{11} & \gamma_2\Delta p^i_{21} & \cdots & \gamma_{m_i}\Delta p^i_{m_i1}\\
\gamma_1\Delta p^i_{12} & \gamma_2\Delta p^i_{22} & \cdots & \gamma_{m_i}\Delta p^i_{m_i2}\\
\vdots & \vdots & \ddots & \vdots\\
\gamma_1\Delta p^i_{1m_i} & \gamma_2\Delta p^i_{2m_i} & \cdots & \gamma_{m_i}\Delta p^i_{m_im_i}\\
\end{matrix} \right\rVert_{\infty}.
\end{align*}

By applying the result in Lemma \ref{lemma:4.1} into \eqref{eqn:15-1}, it is guaranteed that $\rho\Big(({P^i}^{\top}\otimes I)\text{\textnormal{diag}}(\hat{W}_{j}^{i}\otimes \hat{W}_{j}^{i})\Big) < 1$, 
if the following condition
\begin{align*}
\displaystyle\sum_{r=1}^{m_i}\alpha_r|\Delta p^i_{rs}| + \sum_{r=1}^{m_i}\alpha_r\bar{p}^i_{rs} < 1,\quad \begin{matrix}
\forall s=1,2,\hdots,m_i,\\
\end{matrix}
\end{align*}
where $\alpha_r\triangleq ||\hat{W}_r^i\otimes \hat{W}_r^i||_{\infty}$, is satisfied.

Therefore, \eqref{eqn:5} is mean square stable by Theorem \ref{theorem:4.1} if it is guaranteed that
\begin{align*}
\displaystyle\sum_{r=1}^{m_i}\alpha_r|\Delta p^i_{rs}|  < \beta_s,\quad \begin{matrix}
\forall s=1,2,\hdots,m_i,\\
\forall i=1,2,\hdots,N,
\end{matrix}
\end{align*}
where $\beta_s \triangleq 1 - \sum_{r=1}^{m_i}\alpha_r\bar{p}^i_{rs}$.
\end{proof}

\section{Examples}
\subsection{Stability Analysis for $N$ Inverted Pendulum System}
Consider $N$ inverted pendulum system, which are physically interconnected by linear springs \cite{wang2008event}. The discrete-time subsystem dynamics with communication delays is modeled by
\begin{align*}
&\quad\qquad x_i(k+1) = A_ix_i(k) +
B_iu_i(k) + 
\sum_{\substack{j\in\mathcal{N}_i \\ j\neq i}}
H_{ij}x_j(k^*),
\end{align*}
with subsystem matrices:
\begin{align*}
&A_i = \begin{bmatrix}
1 & \Delta t\\
\left(\frac{g}{l} \textbf{--} \frac{a_iK}{\check{m}l^2}\right)\Delta t & 1
\end{bmatrix},\:
B_i = 
\begin{bmatrix}
0\\
\frac{\Delta t}{\check{m}l^2}
\end{bmatrix},\:
H_{ij} = 
\begin{bmatrix}
0 & 0\\
\frac{h_{ij}K}{\check{m}l^2}\Delta t & 0
\end{bmatrix},
\end{align*}
where $k$ denotes the discrete-time index and $x_i = (x_{i_1},\:x_{i_2})^{\top}\in\mathbb{R}^2$. 
The communication delay is described by the term $k^* = k - \tau$ with the discrete value $\tau$. The meaning of each parameter and its value are given in Table \ref{table_inverted_pendulum}. 

For this system, we consider a state feedback law given by $u_i(k) = K_ix_i(k)$, where $K_i\triangleq \displaystyle\begin{bmatrix}
a_iK-\frac{\check{m}l^2}{4}(8+\frac{4g}{l}), & -3\check{m}l^2\end{bmatrix}$
for the control input $u_i(k^*)$.

\begin{table}[h]
\begin{center}
\caption{Nomenclature for $N$ Inverted Pendulum Dynamics}\label{table_inverted_pendulum}
{\small
  \begin{tabular}{|c|c|c|}\hline
  Definition & Symbol & Value\\\hline
	\begin{tabular}{@{}c@{}}Number of Springs connected to \\$i^{th}$ Pendulum\end{tabular} & $a_i$ & 
	\begin{tabular}{@{}c@{}}1, \:$i=1,N$\\\:2, otherwise\end{tabular}\\
	Interaction term with neighbours & $h_{ij}$ &  0.04,\:\:$\forall j\in\mathcal{N}_i$\\
	Gravity & $g$ & 9.8\\
	Spring Constant & $K$ & 5\\
	Pendulum Mass & $\check{m}$ & 0.5\\
	Pendulum Length & $l$ & 1\\
	Sampling time for discrete-time dynamics & $\Delta t$ & 0.1
  \\\hline
  \end{tabular}
 }
\end{center}
\end{table}

We can rewrite the closed-loop dynamics for this $N$ inverted pendulum system as follows:
\begin{align*}
x_i(k+1) = \sum_{j\in\mathcal{N}_i}A_{ij}x_j(k^*), \text{ where }A_{ij} \triangleq \begin{dcases}
A_{i}+B_iK_i, \text{ if } j=i,\\H_{ij}, \text{ otherwise}.
\end{dcases}
\end{align*}

If there is no communication delay (i.e., $k^*=k$), the dynamics for the entire DNCS is given by \eqref{eqn:4}, where the matrix $\mathcal{A}$ of which structure is also given in \eqref{eqn:4} satisfies $\rho(\mathcal{A}) = 0.9525 < 1$. Therefore, we can assure that the $N$ inverted pendulum system with no communication delays is stable.

Next, we test the stability for this system with random communication delays.
We assume that the communication delay $\tau$ is bounded by $0 \leq \tau \leq \tau_{d} = 1$, i.e., $k^* = \{k,\:k-1\}$, $\forall i=1,2,\hdots,N$. Also, we assume that every communication delays are governed by the Markov process with an initial probability distribution $\pi(0)$ and the Markov transition probability matrix $P$ given by
\begin{eqnarray}
\pi(0) = [1,\:0], \quad
P = \begin{bmatrix}
0.5 & 0.5\\
0.3 & 0.7
\end{bmatrix}.\label{eqn:18}
\end{eqnarray}

For this system, even with $N=100$, the full state model \eqref{eqn:5} has total $q^{N(N-1)} = 2^{100\times 99}$ modes. 
Since this inverted pendulum system has only interconnected terms with neighbors when $j= i\pm 1$, otherwise we have $A_{ij}=0$.
Based on this fact and by excluding these cases (i.e., where $A_{ij}=0$), we can further reduce the mode number to $q^{2(N-1)}=2^{2\times 99}$, which is still large. It is computationally intractable to deal with $2^{2\times 99}$ numbers of matrices to analyze system stability. However, in contrast, the reduce mode model \eqref{eqn:12} has total $\sum_{i=1}^{N}q^{\hat{n}_i(\hat{n}_i-1)} = 98\times(2^{3\times 2}) + 2\times (2^{2\times 1}) = 6280$ modes. 
Furthermore, the proposed method fully maximizes its own advantage to reduce the mode numbers by considering the symmetric property between agents, which cannot be implemented on the full state model.
Since subsystems are symmetric for $\forall i=2,3,\hdots,N-1$ and for $\forall i=1,N$, we only need to check the stability condition for these two cases. Taking into account the interconnection link (i.e., the case where $A_{ij}\neq 0$), the symmetric structure results in total $2^{2\times (3-1)}+2^{2\times (2-1)} = 20$ modes, which is drastically reduced when compared to $2^{2\times 99}$ numbers of modes.

\begin{figure}[!ht]
\centering
\includegraphics[scale=0.85]{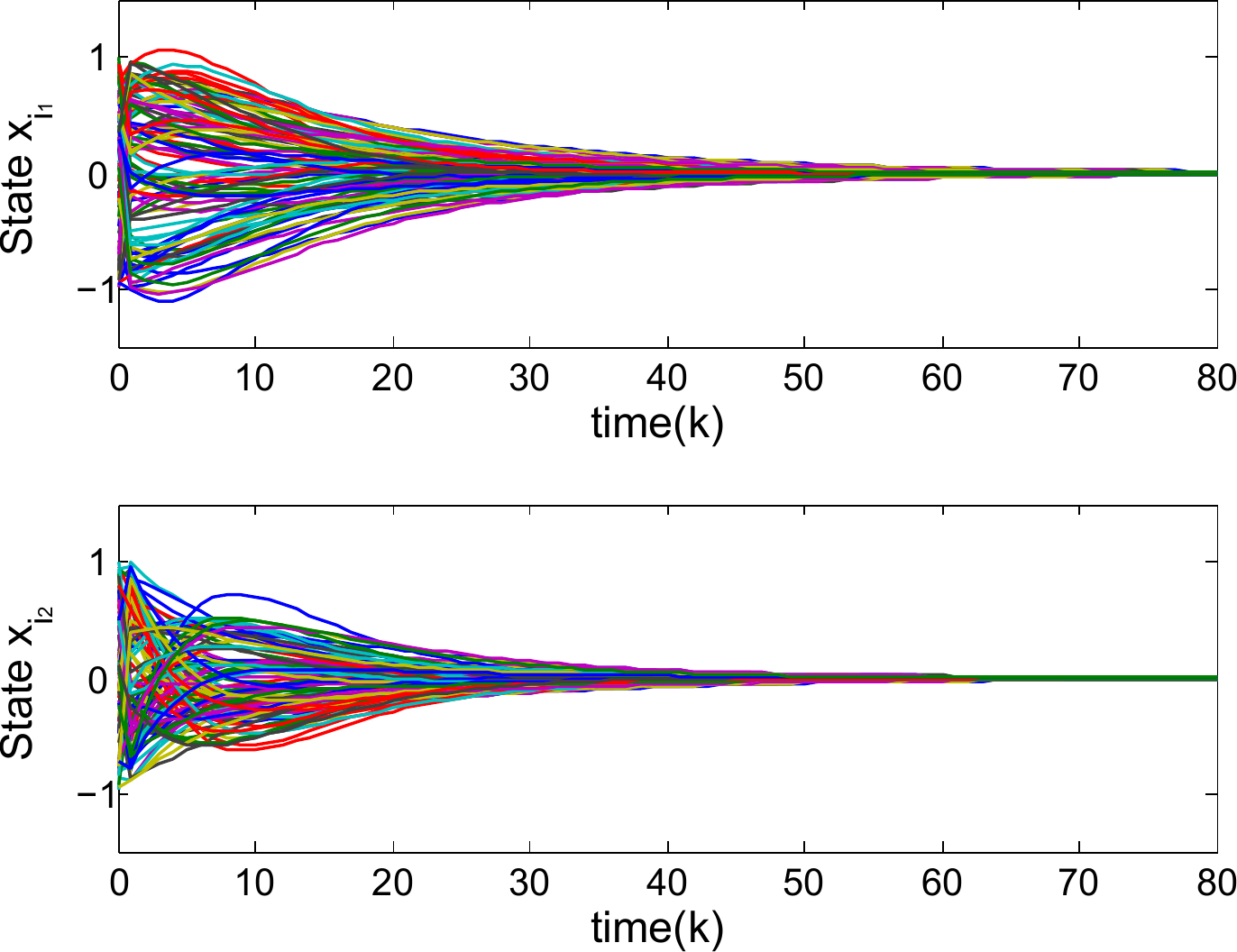}
\caption{State trajectories of each agent for the $N$ inverted pendulum system with the Markovian communication delays. Initial conditions are randomly generated for all states $x_i$, $\forall i=1,2,\hdots,100$.}\label{fig:3}
\end{figure}

The spectral radius for $i=2,3,\hdots,99$ is computed by $\rho\Big( ({P^i}^{\top}\otimes I)$ $\text{\textnormal{diag}}(\hat{W}_j^i\otimes \hat{W}_j^i)\Big)=0.8864 < 1$, where $P^i=(P\otimes P\otimes P\otimes P)$. For $i=1$ and $N$, we have $\rho\Big( ({P^i}^{\top}\otimes I)\text{\textnormal{diag}}(\hat{W}_j^i\otimes \hat{W}_j^i)\Big)=0.8682 < 1$, where $P^i = (P\otimes P)$. Consequently, the $N$ inverted pendulum system is stable in the mean square sense according to Theorem \ref{theorem:4.1}. The state trajectory plot also supports this result, as shown in Fig. \ref{fig:3}.
For this simulation, initial condition was assumed to be uniformly distributed in $[-1,1]$, and was generated by manipulating the MATLAB$^{{\tiny \mycirc{R}}}$ command \texttt{rand(...)} that generates uniformly distributed pseudo random numbers between $0$ and $1$.

\subsection{Stability Bound for Uncertain Markov transition probability matrix}
In order to solve the LP to estimate the bound for uncertain Markov transition probability matrix, we used MATLAB$^{{\tiny \mycirc{R}}}$ with CVX\cite{grant2008cvx}, a Matlab-based software for convex optimization.
\subsubsection{Scalar system}
Although the proposed method to estimate maximum bound for uncertain Markov transition probability matrix is developed for the large-scale DNCS, it is also applicable to general MJLS. We adopted a following example, introduced in \cite{karan2006transition} to compare the performance of maximum bound estimation.

Consider the following MJLS that has two modes with scalar discrete-time dynamics. 
\begin{align*}
&x(k+1) = a_{\sigma(k)}x(k),\quad\sigma(k)\in\{1,2\},\\
&a_1 = 1/2,\quad a_2=5/4.
\end{align*}
The Markov transition probability matrix has the form of $P=\bar{P}+\Delta P$, where
\begin{align*}
\bar{P} = \begin{bmatrix}
0.4 & 0.6\\
0.5 & 0.5
\end{bmatrix},\:
\Delta{P} = \begin{bmatrix}
\Delta p_{11} & \Delta p_{12}\\
\Delta p_{21} & \Delta p_{22}
\end{bmatrix},\:\sum_{j=1}^{2}\Delta p_{ij} = 0,\quad\forall i=1,2.
\end{align*}

After applying the two steps proposed in this paper, we obtained the maximum bound $\varepsilon_1=0.4,\:\varepsilon_2=0.02$ whereas \cite{karan2006transition} gives the value as $\varepsilon_1=\varepsilon_2=0.021$, which is more conservative. For stability check, among all possible scenarios with $|\Delta p_{rs}| \leq \varepsilon_r$, $\forall r,s=1,2$, we have $\max\rho\Big((P^{\top}\otimes I)\textnormal{diag}(a_j\otimes a_j)\Big) = 1$, which is a marginal value for stability. Hence, the system is stable with obtained maximum bound that is more relaxed than \cite{karan2006transition}.

\subsubsection{The $N$ Inverted Pendulum System}
Recalling the $N$ inverted pendulum system, we assume that the Markov transition probability matrix $P^i $ has uncertainty $\Delta P^i$ that satisfies $\Delta P^i = P^i - \bar{P}^i$. The nominal matrix $\bar{P}^i$ is given by $\bar{P}^i = (\bar{P}\otimes \bar{P})$ for $i=1,N$ and $\bar{P}^i = (\bar{P}\otimes \bar{P}\otimes \bar{P} \otimes \bar{P})$ for $i=2,3,\hdots,N-1$, where $\bar{P}$ has a same structure with the transition probability matrix given in \eqref{eqn:18}.

The feasible solution with the LP provides the maximum bound $\varepsilon^i = [\varepsilon_1^i,\:\varepsilon_2^i,\:\hdots\varepsilon_{16}^i]=10^{-3}\times [4.9,\:0.9,\:0.9,$ $\:0.8,\:0.9,\:0.8,\:0.8,\:6.9,\:0.9,\:0.8,\:0.8,$ $\:6.9,\:0.8,\:6.9,\:6.9,\:13.5]$, $\forall i=2,3,\hdots,N-1$.
For $i=1$ and $N$, we obtained $\varepsilon^i = [\varepsilon_1^i,\:\varepsilon_2^i,\:\varepsilon_3^i,\:\varepsilon_4^i] = [0.01,\:0.01,\:0.01,\:0.03]$.
Therefore, we can assure that $N$ inverted pendulum system is mean square stable if the uncertainty in the Markov transition probability matrix is within the bound such that $|\Delta p^i_{rs}| \leq \varepsilon^i_r$, $\forall r$ and $\forall i=1,2,\hdots,N$.

\section{Conclusions}
This paper studied the mean square stability of the large-scale DNCSs. Since the number of modes in such systems is extremely large, current stability analysis tools are intractable. To avoid this scalability problem, we provided a new analysis framework, which incorporates a reduced mode model that scales linearly with respect to the number of subsystems. Additionally, we presented a new method to estimate bounds for uncertain Markov transition probability matrix for which system stability is guaranteed. We showed that this method is less conservative than those proposed in the literature. The validity of the proposed methods were verified using an example based on interconnected inverted pendulums.

\section{Acknowledgements}
This research was supported by the National Science Foundation award \#1349100, with Dr. Almadena Y. Chtchelkanova as the program manager.\\

\footnotesize
\bibliographystyle{unsrt}
\bibliography{DNCS2014}

\end{document}